\documentclass[runningheads]{llncs}

\usepackage{amsmath, amsfonts, amssymb, bm}
\usepackage{graphicx}
\usepackage{booktabs}
\usepackage{mathtools}
\usepackage{hyperref}
\usepackage{algorithm}
\usepackage{algorithmic}
\usepackage{physics}
\usepackage{dsfont}
\usepackage{enumitem}
\usepackage{subcaption}
\usepackage{xcolor}

\newtheorem{assumption}{Assumption}
\DeclareMathOperator{\conv}{conv}

\hypersetup{
    colorlinks=true,
    linkcolor=blue,
    citecolor=blue,
    urlcolor=blue
}
\usepackage[margin=1in]{geometry}
\begin{document}

\title{On the Expressive Power and Capacity of Quantum Data Reuploaders}

\titlerunning{Expressivity and Capacity of Data Reuploaders}

\author{
  Lord Sen \and
  Shyamapada Mukherjee
}

\authorrunning{Sen and Mukherjee}

\institute{
  National Institute of Technology, Rourkela, Rourkela, India \\
  \email{lordsen3008@gmail.com}
  \email{mukherjees@nitrkl.ac.in}
}

\maketitle

\begin{abstract}
Data reuploading---repeatedly encoding classical inputs at multiple layers of a
parameterized quantum circuit---is a central mechanism for enhancing the
expressivity of quantum neural networks (QNNs). Building on the Fourier-theoretic
framework of Schuld et al.~\cite{schuld2021fourier}, we study the frequency
spectrum of the general multivariate, $R$-reupload architecture (QDR).
(i) We make explicit that the accessible frequencies form the set
$K_R=S^{(R)}-S^{(R)}$, which equals the $R$-fold sumset $T^{(R)}$ of the
difference set $T=S-S$ of the encoding spectrum, and we characterise it exactly:
$K_R=\{k\in\mathcal L:\ \ell_T(k)\le R\}$, where $\mathcal L=\mathrm{span}_{\mathbb Z}(T)$
is an integer lattice and $\ell_T$ is word length in $T$.
(ii) We show that the \emph{directions} of the spectrum stabilise already at $R=1$
(the real span of $K_R$ equals that of $\mathcal L$), whereas the spectrum itself
never saturates: the maximal frequency grows linearly in $R$ and the number of
frequencies grows polynomially, as $\Theta(R^{m})$ with $m=\mathrm{rank}\,\mathcal L$.
(iii) Using the spectrum size, we give a Rademacher-complexity bound on the
capacity of QDRs of order $\sqrt{|K_R|/n}$.
(iv) We report experiments on function approximation and classification with
train/test splits, several seeds and parameter-matched classical baselines. The
simulated spectra match the theory exactly. A single-qubit QDR beats a matched MLP on
several smooth or periodic univariate targets and on the spiral task, but it is
clearly worse on multivariate regression ($d\ge4$), where it is
highly sensitive to the input-encoding scale; it beats the MLP only on equations with $d\le3$ inputs.
Code: \url{https://github.com/Lord3008/QDRs}
\end{abstract}


\section{Introduction}

Quantum computing and circuits have found applications in many
areas~\cite{Shafique,sen2025symmetric,Nakata,sen2025quantum}, and parameterized
quantum circuits (PQCs) are central to near-term quantum machine
learning~\cite{SchuldReview,Cincio}. While quantum states inhabit exponentially
large Hilbert spaces, current hardware is constrained by noise, circuit depth
and connectivity. A popular strategy in this regime is \emph{data reuploading}:
repeatedly encoding classical inputs into a circuit to generate nonlinear
features. A general Quantum Data Reuploader (QDR)~\cite{PerezSalinas,Wang2025}
with $R$ encoding gates has the form
\[
U(\theta, x) = W_{R+1}(\theta_{R+1})\, V(x)\, W_{R}(\theta_{R}) \cdots V(x)\, W_1(\theta_1),
\]
where the $W_i$ are trainable unitaries and $V(x)$ encodes the
data, e.g.\ as a product $V(x)=\prod_j e^{-ix_jH_j}$ or, for commuting generators, as $e^{-i\sum_j x_jH_j}$ (see Section~\ref{subsec:fourier}). The model output is
\[
f_\theta(x) = \langle 0^{\otimes n} | U^\dagger(\theta,x) M U(\theta,x) | 0^{\otimes n} \rangle .
\]

Schuld et al.~\cite{schuld2021fourier} showed that such models are trigonometric
polynomials whose frequency content is fixed by the eigenvalues of the encoding
Hamiltonians and the number of encodings. A spectrum that generates a
full-rank lattice is \emph{necessary} for a model to represent arbitrary
$d$-dimensional Fourier structure (Section~\ref{subsec:structure}); it is not by itself sufficient, since
the trainable blocks must also be able to realise the coefficients. Asymptotic universality
of single-qubit reuploading models is established
in~\cite{PerezSalinasOne,Yu2022}. The frequency picture has been developed
further in, e.g., \cite{Casas2023,Barthe2023}. Universality is therefore an
\emph{asymptotic} statement in $R$; at finite depth the capacity is limited by the
spectrum $K_R$ and by the number of trainable parameters, which is the regime our
experiments and results address.

\paragraph{Contributions.}
\begin{enumerate}
    \item We restate the Fourier representation for the general multivariate,
    $R$-reupload architecture in a single explicit lemma (Section~\ref{subsec:fourier}).
    This is a direct adaptation of~\cite{schuld2021fourier} (their Sec.~II and
    Appendix~A) and is not claimed as new theory; it fixes notation.
    \item We give an exact description of the spectrum (Theorem~\ref{thm:structure}):
    $K_R$ is the ball of radius $R$ in the word metric of $T$ on the lattice $\mathcal L$.
    Consequences: (a) the spectrum's directions are fixed at $R=1$;
    (b) every fixed frequency in $\mathcal L$ is reached at a finite, exactly computable depth;
    (c) the spectrum never saturates and has polynomial growth $\Theta(R^m)$.
    These statements are elementary consequences of sumset structure, and we do not
    claim priority over the frequency-spectrum analyses in the literature
    (e.g.~\cite{schuld2021fourier,Casas2023,Barthe2023}); we make them explicit
    because they govern the finite-depth behaviour of QDRs.
    \item A capacity bound in terms of $|K_R|$ (Proposition~\ref{prop:capacity})
    and a discussion of how parameter count and coefficient constraints interact with
    the spectrum (Section~\ref{subsec:coeff}).
    \item Experiments on function approximation and classification
    (Section~\ref{sec:exp}) with train/test splits, multiple seeds, parameter-matched
    MLP baselines, a numerical verification of the predicted spectra, and an explicit account of
    what they do and do not show.
\end{enumerate}

\section{Theoretical Analysis of QDRs}

\subsection{Fourier Representation of Data Reuploading Models}
\label{subsec:fourier}

We consider the $R$-reupload model
\begin{equation}
    U(\theta,x)
    = W_{R+1} V(x) W_{R} V(x) \cdots W_2 V(x) W_1,
    \label{eq:U}
\end{equation}
containing $R$ encoding gates and $R+1$ trainable blocks, with
$f_\theta(x)=\langle 0^{\otimes n}|U^\dagger M U|0^{\otimes n}\rangle$ for a
Hermitian observable $M$.

The only structural property of the encoding we need is the following.

\begin{assumption}[Finite phase decomposition]
\label{ass:decomp}
There is a finite set $S\subset\mathbb R^d$ and $x$-independent operators
$\{P_s\}_{s\in S}$ with
\begin{equation}
    V(x)=\sum_{s\in S} e^{-i x\cdot s}\,P_s .
    \label{eq:Vdecomp}
\end{equation}
\end{assumption}

This holds in the two standard settings:
\begin{itemize}
    \item \emph{Commuting generators.} If $[H_i,H_j]=0$ for all $i,j$ (e.g.\ Pauli-$Z$
    type encodings, including several on the same qubit), the $H_j$ have joint spectral
    projectors $\Pi_s$; take $P_s=\Pi_s$ and $S$ the \emph{joint spectrum}
    $\{s:\Pi_s\neq0\}$. Note that $S$ can be a proper subset of
    $\sigma(H_1)\times\dots\times\sigma(H_d)$ (Example~\ref{ex:joint}).
    \item \emph{Product encoding.} For $V(x)=\prod_{j=1}^d e^{-ix_jH_j}$, expanding each
    factor spectrally gives~\eqref{eq:Vdecomp} with $P_s=\Pi_{1,a_1}\cdots\Pi_{d,a_d}$
    (an ordered product; not a projector in general) and $S\subseteq\prod_j\sigma(H_j)$.
    Commutation is not needed.
\end{itemize}
If the $H_j$ do not commute, the ``sum'' encoding $e^{-i\sum_j x_jH_j}$ is in general
not a finite trigonometric polynomial in $x$, and Assumption~\ref{ass:decomp} may fail.
Write $\Lambda_{\max}:=\max_{s\in S}\|s\|_\infty$.

\begin{lemma}[Fourier representation, adapted from~\cite{schuld2021fourier}]
\label{lem:fourier}
Under Assumption~\ref{ass:decomp}, the output of the $R$-reupload model~\eqref{eq:U} is a
finite Fourier series
\begin{equation}
    f_\theta(x)
    = \sum_{k \in K_R} c_k(\theta)\, e^{i k \cdot x},
    \qquad
    K_R \;:=\; S^{(R)} - S^{(R)},
    \label{eq:KR}
\end{equation}
where $S^{(R)} := \{s_1+\cdots+s_R \mid s_i\in S\}$ is the $R$-fold Minkowski sum.
Moreover $c_{-k}=\overline{c_k}$ (so $f_\theta$ is real), and
\begin{equation}
    \| k \|_\infty \;\le\; 2R \Lambda_{\max}
    \qquad \text{for every } k\in K_R .
    \label{eq:bound}
\end{equation}
Thus $K_R$ is the set of frequencies the architecture \emph{may} carry; whether a
given $c_k(\theta)$ can be non-zero depends on the trainable blocks.
\end{lemma}

\begin{proof}
Insert~\eqref{eq:Vdecomp} into each of the $R$ occurrences of $V(x)$ in~\eqref{eq:U}:
\begin{equation}
    U(\theta,x) = \sum_{\vec a \in S^{R}}
        e^{-i x\cdot \Lambda(\vec a)}\, B_{\vec a}(\theta),
    \qquad
    \Lambda(\vec a) := \sum_{t=1}^R a_t \in S^{(R)},
    \label{eq:Uexp}
\end{equation}
with $B_{\vec a}(\theta) := W_{R+1} P_{a_R} W_{R}P_{a_{R-1}}\cdots P_{a_1} W_1$
independent of $x$. Taking the adjoint of~\eqref{eq:Uexp},
$U^\dagger=\sum_{\vec b}e^{+ix\cdot\Lambda(\vec b)}B_{\vec b}^\dagger$; the phases of
$U^\dagger$ have the opposite sign to those of $U$, which is the origin of the
difference in~\eqref{eq:KR}. Hence
\begin{equation}
    f_\theta(x)
    = \sum_{\vec a,\vec b \in S^{R}}
        e^{\,i x\cdot(\Lambda(\vec b)-\Lambda(\vec a))}\,
        \langle 0^{\otimes n}| B_{\vec b}^\dagger\, M\, B_{\vec a} |0^{\otimes n}\rangle .
\end{equation}
Grouping terms by $k:=\Lambda(\vec b)-\Lambda(\vec a)$ gives~\eqref{eq:KR} with
\[
c_k(\theta) := \sum_{\Lambda(\vec b)-\Lambda(\vec a)=k}
\langle 0^{\otimes n}| B_{\vec b}^\dagger M B_{\vec a}|0^{\otimes n}\rangle ,
\]
and, since $\Lambda(\vec a),\Lambda(\vec b)$ range independently over $S^{(R)}$, $k$
ranges over $S^{(R)}-S^{(R)}$. Reality of $f_\theta$ ($M$ Hermitian) together with
uniqueness of Fourier coefficients gives $c_{-k}=\overline{c_k}$. For the bound,
$\|\Lambda(\vec a)\|_\infty\le R\Lambda_{\max}$ and likewise for $\vec b$, so
$\|k\|_\infty\le 2R\Lambda_{\max}$.
\end{proof}

\begin{remark}[Relation to~\cite{schuld2021fourier}]
Lemma~\ref{lem:fourier} is not a new result. For $d=1$ it recovers the
sequential-reuploading spectrum of \cite[Eqs.~(10),(12)]{schuld2021fourier}; for $R=1$
and one Hamiltonian per variable it recovers their multivariate single-layer result
\cite[Eq.~(A10)]{schuld2021fourier}. The lemma records that these compose, and fixes the
frequency set $K_R$ analysed below.
\end{remark}

\subsection{Structure of the Spectrum}
\label{subsec:structure}

\begin{assumption}[Commensurate spectra]
\label{ass:commensurate}
$S\subset\mathbb Z^d$ (equivalently, $S\subset\mathbb Q^d$ after an overall rescaling
of $x$). This holds for Pauli-type encodings (eigenvalues $\pm1$) and any rational spectrum.
\end{assumption}

\begin{remark}
Without Assumption~\ref{ass:commensurate} the group $\mathcal L$ below can be
non-discrete: for $S=\{0,1,\sqrt2\}\subset\mathbb R$, $\mathcal L=\mathbb Z+\sqrt2\mathbb Z$ is dense in $\mathbb R$
and has $\mathbb Z$-rank $2>d=1$. Then $\mathcal L\cap B_\infty(\rho)$ is infinite, and
statements such as $\mathrm{rank}\,\mathcal L\le d$ and finiteness of windows fail.
We therefore assume Assumption~\ref{ass:commensurate} from here on.
\end{remark}

Define the \emph{difference set} $T:=S-S\subset\mathbb Z^d$, so $0\in T$ and $T=-T$, and
\[
\Delta:=\max_{t\in T}\|t\|_\infty\ \ (\le 2\Lambda_{\max}),\qquad
\mathcal L:=\mathrm{span}_{\mathbb Z}(T),\qquad m:=\rank\mathcal L\le d .
\]
For $k\in\mathcal L$ let $\ell_T(k)$ be the least number of elements of $T$ summing to $k$
(the \emph{word length} of $k$ with respect to $T$; finite because $T$ generates $\mathcal L$
and $T=-T$).

\begin{lemma}[$K_R$ is the $R$-fold sumset of $T$]
\label{lem:KR=TR}
For every $R\ge 1$, $K_R=S^{(R)}-S^{(R)}=T^{(R)}:=\{t_1+\cdots+t_R \mid t_i\in T\}$.
Hence $K_R\subseteq K_{R+1}$, each $K_R$ is finite and symmetric, and
\begin{equation}
    K_R=\{k\in\mathcal L:\ \ell_T(k)\le R\}.
    \label{eq:wordball}
\end{equation}
\end{lemma}
\begin{proof}
If $k=u-v$ with $u=\sum_i s_i$, $v=\sum_i s_i'$ then $k=\sum_i(s_i-s_i')\in T^{(R)}$;
conversely $\sum_i(s_i-s_i')=\sum_i s_i-\sum_i s_i'$. Since $0\in T$, any sum of $r\le R$
elements of $T$ can be padded with zeros to exactly $R$ terms, so $K_R\subseteq K_{R+1}$
and~\eqref{eq:wordball} follows; symmetry follows from $T=-T$.
\end{proof}

Every reachable frequency lies in $\mathcal L=\mathrm{span}_{\mathbb Z}(S-S)$, and not in
the possibly larger $\mathrm{span}_{\mathbb Z}(S)$. For $S=\{\pm1\}$, $\mathrm{span}_{\mathbb Z}(S)=\mathbb Z$,
but no $K_R$ contains an odd integer.

\begin{theorem}[Structure of the spectrum of a QDR]
\label{thm:structure}
Under Assumption~\ref{ass:commensurate}, with $T,\Delta,\mathcal L,m,\ell_T$ as above:
\begin{enumerate}
    \item \textbf{(Exhaustion and exact reach depth.)} $\bigcup_R K_R=\mathcal L$. A frequency
    $k\in\mathcal L$ belongs to $K_R$ if and only if $R\ge\ell_T(k)$. For every $\rho>0$,
    \[
    K_R\supseteq\mathcal L\cap B_\infty(\rho)\iff R\ge R_{\rm win}(\rho):=\max_{k\in\mathcal L\cap B_\infty(\rho)}\ell_T(k),
    \]
    and $R_{\rm win}(\rho)\ge\lfloor\rho/\Delta\rfloor$.
    \item \textbf{(Confinement.)} $K_R\subseteq\mathcal L\cap R\cdot\conv(T)\subseteq\mathcal L\cap B_\infty(R\Delta)\subseteq \mathcal L\cap B_\infty(2R\Lambda_{\max})$.
    \item \textbf{(Directions stabilise at $R=1$.)} For every $R\ge1$,
    $\mathrm{span}_{\mathbb R}(K_R)=\mathrm{span}_{\mathbb R}(\mathcal L)=\mathrm{span}_{\mathbb R}(T)=:V$, a subspace of dimension $m$.
    Consequently $f_\theta(x+v)=f_\theta(x)$ for all $v\in V^\perp$: if $m<d$ the model depends
    on $x$ only through its projection onto $V$, for every depth.
    \item \textbf{(No saturation of the spectrum.)} If $T\ne\{0\}$, then $K_R\subsetneq K_{R+1}$ for
    every $R$; in particular $\max_{k\in K_R}\|k\|_\infty$ grows linearly in $R$ and $|K_R|\to\infty$.
    Moreover, by Khovanskii's theorem on sumsets~\cite{Khovanskii,Nathanson}, $|K_R|$ agrees for large $R$
    with a polynomial of degree $m$ in $R$, with leading coefficient $\mathrm{vol}_m(\conv T)/\mathrm{covol}(\mathcal L)$
    (cited, not reproved here); so $|K_R|=\Theta(R^m)$.
\end{enumerate}
\end{theorem}

\begin{proof}
(1) Let $k\in\mathcal L$. Then $k=\sum_{t\in T}n_t t$ with $n_t\in\mathbb Z$; replacing $n_t t$ by
$|n_t|$ copies of $-t\in T$ when $n_t<0$ writes $k$ as a sum of finitely many elements of $T$, so $\ell_T(k)<\infty$
and $k\in K_{\ell_T(k)}$. The equivalence $k\in K_R\iff R\ge\ell_T(k)$ is~\eqref{eq:wordball}. The window statement
follows because $\mathcal L\cap B_\infty(\rho)$ is finite (a discrete group intersected with a bounded set), and
for the lower bound, take $t_0\in T$ with $\|t_0\|_\infty=\Delta$ and $j=\lfloor\rho/\Delta\rfloor$; then $jt_0\in\mathcal L\cap B_\infty(\rho)$, and since by part~(2) a sum of $R$ elements of $T$ has sup-norm at most $R\Delta$, $\ell_T(jt_0)\ge j$.

(2) A sum of $R$ elements of $T$ lies in $R\cdot\conv(T)$ (divide by $R$); each element has $\|t\|_\infty\le\Delta$.

(3) $K_1=T\subseteq K_R\subseteq\mathcal L\subseteq\mathrm{span}_{\mathbb R}(T)$, so all spans coincide.
The dimension of $\mathrm{span}_{\mathbb R}(T)$ equals $\rank\mathcal L=m$ because $T\subset\mathbb Z^d$.
If $v\perp V$ then $k\cdot v=0$ for every $k\in K_R$, so each term $e^{ik\cdot(x+v)}=e^{ik\cdot x}$.

(4) Pick $t_0\in T\setminus\{0\}$ and $\phi(y):=t_0\cdot y$, so $h:=\max_{t\in T}\phi(t)\ge\|t_0\|^2>0$.
Then $\max_{K_R}\phi=Rh$ (attained by $R$ copies of a maximiser), which strictly increases with $R$.
So $K_{R+1}\not\subseteq K_R$.
\end{proof}

\begin{remark}[What ``saturation'' does and does not mean]
What stabilises with depth is the \emph{direction structure} (the subspace $V$ and the lattice
$\mathcal L$), and it does so at $R=1$. The set of frequencies itself never saturates
(part 4): each extra reupload adds new, larger frequencies along the lattice. Part 1 gives
the exact depth $R_{\rm win}(\rho)$ at which a whole window of $\mathcal L$ is populated; it grows at
least like $\lfloor\rho/\Delta\rfloor$, i.e.\ linearly in $\rho$. For $m\ge2$, $\mathcal L$ has infinitely many primitive directions and $K_R$ is
always finite, so no finite depth reaches all of them. We obtain no upper bound on $R_{\rm win}$ beyond
its exact definition via $\ell_T$, which is computable for any given $S$.
\end{remark}

\begin{example}[Single qubit, $V(x)=e^{-ixZ}$]
$S=\{\pm1\}$, $T=\{-2,0,2\}$, $\mathcal L=2\mathbb Z$, $m=1$, $\ell_T(2j)=|j|$, and
$K_R=\{-2R,-2R+2,\dots,2R\}$, so $|K_R|=2R+1=\Theta(R)$ and $\max\|k\|=2R$, attaining the bound of Theorem~\ref{thm:structure}(2)
with $\Delta=2$. The generator $\pm2$ already lies in $K_1$, so the \emph{directions} are
complete at $R=1$, but the model keeps gaining higher harmonics as $R$ grows.
\end{example}

\begin{example}[Shifted spectrum]
For $S=\{0,3\}$, $T=\{0,\pm3\}$ and $K_R=\{-3R,\dots,3R\}\cap3\mathbb Z$. Here
$\Delta=3<2\Lambda_{\max}=6$: the bound $2R\Lambda_{\max}$ is not tight while $R\Delta$ is.
\end{example}

\begin{example}[Joint spectrum need not be a product]
\label{ex:joint}
Two inputs encoded on two qubits with $H_1=Z\otimes I$, $H_2=I\otimes Z$ give
$S=\{\pm1\}^2$, $T=\{-2,0,2\}^2$, $\mathcal L=(2\mathbb Z)^2$, $m=2$, $|K_R|=(2R+1)^2$.
If instead both inputs are encoded by $Z$ on the same qubit ($H_1=H_2=Z$), the joint
spectrum is $S=\{(1,1),(-1,-1)\}\subsetneq\{\pm1\}^2$, $T=\{0,\pm(2,2)\}$, $m=1<d=2$,
and by Theorem~\ref{thm:structure}(3) $f_\theta$ depends on $(x_1,x_2)$ only through
$x_1+x_2$, at any depth.
\end{example}

\subsection{Coefficients, Parameters and Capacity}
\label{subsec:coeff}

The spectrum $K_R$ restricts \emph{which} frequencies can appear; it does not by itself
determine the expressivity of the model. The coefficients $c_k(\theta)$ are constrained by the
architecture: they satisfy $c_{-k}=\overline{c_k}$, and they are functions of only
$P=\sum_i\dim\theta_i$ trainable parameters (about $(R+1)p$ for blocks with $p$ parameters).

\begin{remark}[Parameter count versus spectrum size]
A real-valued trigonometric polynomial with symmetric spectrum $K_R\ni0$ has exactly $|K_R|$
real degrees of freedom ($c_0$ real, plus $(|K_R|-1)/2$ complex conjugate pairs). If $P<|K_R|$, the
parameter-to-function map cannot fill an open subset of that space. Since $P$ grows linearly in $R$
and $|K_R|=\Theta(R^m)$, for $m\ge2$ the spectrum outgrows the parameter budget, and it is the
number of trainable parameters, not $K_R$, that limits the model. For a single qubit ($m=1$) with
general $SU(2)$ blocks ($p=3$), $P=3(R+1)>2R+1=|K_R|$, so parameter count is not the binding constraint,
but the constraint $|f_\theta|\le\|M\|$ and the structure of the reachable coefficients still are
(cf.~\cite{Yu2022}).
\end{remark}

\begin{proposition}[Spectral capacity bound]
\label{prop:capacity}
Let $\mathcal F_R=\{f_\theta:\theta\}$ be the class of $R$-reupload models with $\|M\|_{\rm op}\le B$, under
Assumptions~\ref{ass:decomp} and~\ref{ass:commensurate}. For any $x_1,\dots,x_n$ the empirical Rademacher complexity satisfies
\[
\widehat{\mathfrak R}_n(\mathcal F_R)\ \le\ B\sqrt{\frac{|K_R|}{n}}\ \le\ B\sqrt{\frac{(2R\Delta+1)^d}{n}},
\qquad |K_R|=\Theta(R^m).
\]
Consequently, for an $\mathrm{Lip}$-Lipschitz loss with values in $[0,C]$ and i.i.d.\ data, with probability
$\ge1-\delta$ every $f\in\mathcal F_R$ satisfies
\[
\text{risk}(f)\le\widehat{\text{risk}}(f)+2\,\mathrm{Lip}\,B\sqrt{|K_R|/n}+C\sqrt{\ln(1/\delta)/(2n)} .
\]
\end{proposition}
\begin{proof}
Write $f_\theta(x)=\langle c,\varphi(x)\rangle$ with $\varphi(x)=(e^{-ik\cdot x})_{k\in K_R}$ (up to
conjugation conventions), so $\|\varphi(x)\|_2^2=|K_R|$. All frequencies lie in $\mathbb Z^d$, so
$f_\theta$ is $2\pi$-periodic in each coordinate; by Parseval on $[0,2\pi)^d$,
$\|c\|_2^2=(2\pi)^{-d}\int|f_\theta|^2\le\|M\|_{\rm op}^2\le B^2$. Then
$\mathbb E_\sigma\sup_c\sum_i\sigma_i\langle c,\varphi(x_i)\rangle\le B\,\mathbb E\|\sum_i\sigma_i\varphi(x_i)\|_2
\le B\sqrt{n|K_R|}$. The size bounds follow from Theorem~\ref{thm:structure}(2),(4), and the
generalization bound is the standard Rademacher bound with Talagrand contraction~\cite{MohriBook}.
\end{proof}

The bound is deliberately spectral: it does not use the parameter count $P$, and can therefore be loose for
small $P$. It does capture the qualitative trade-off: capacity grows like $R^{m/2}$, and the sample size needed for
uniform generalization scales with $|K_R|$. It is an upper bound only; we do not claim it is tight or predictive of
the experiments below. In our experiments the train/test gaps are small (Section~\ref{sec:exp}), so the models are in
an under-fitting regime where this uniform bound is not the relevant constraint.

\subsection{Discussion: Depth, Spectrum and Plateaus}

Two facts must be kept apart. First, the \emph{direction structure} of the spectrum is complete at $R=1$
(Theorem~\ref{thm:structure}(3)), so choosing the encoding Hamiltonians (which fixes $\mathcal L$
and $m$) matters more than raw depth: a lattice of rank $m<d$ cannot represent arbitrary $d$-dimensional
Fourier structure at any depth. Second, at fixed encoding the spectrum keeps growing with $R$, so the spectral
picture does \emph{not} by itself predict a plateau of accuracy in $R$. Any plateau must therefore be attributed to other
factors, such as the limited set of coefficient vectors reachable by the trainable blocks, optimization difficulty, or
the \emph{scale of the inputs relative to the available frequencies}. Our controlled experiments (Section~\ref{sec:exp})
bear on the last factor: the input scale (i.e.\ how many periods of the model the data domain covers) changes the
error by up to an order of magnitude, far more than the depth does. %We did not run experiments that separate coefficient constraints from optimization difficulty.

\section{Experiments}
\label{sec:exp}

Every table reports mean$\pm$standard deviation over random initialisation seeds
(5 seeds for the univariate and classification experiments, 3 for the Feynman experiments); the data set and the train/test split
are fixed across seeds. All errors are measured on a held-out test set unless labelled ``train''.
The experiments are small, single-qubit, and were run with limited hyper-parameter tuning; they illustrate the behaviour of QDRs and test
the spectral statements of Section~2 only qualitatively.

\paragraph{Setup.}
\emph{Architecture.} One qubit, $n=1$. Model~\eqref{eq:U} with $R+1$ trainable general $SU(2)$ blocks
$W_i=R_z(\cdot)R_y(\cdot)R_z(\cdot)$ (3 parameters each), so $P=3(R+1)$ and the parameter counts
$9,15,18,21,24,27,30,36,45$ correspond to $R=2,4,5,6,7,8,9,11,14$. Observable $M=Z$, input state $|0\rangle$.
\emph{Encoding.} For $d$ inputs, $V(x)=\prod_{j=1}^d e^{-i\,\sigma\,x_jP_j}$ (ordered product) with $P_j\in\{X,Y,Z\}$ cycling
with $j$, where $x_j\in[-\pi/2,\pi/2]$ is the affine rescaling of input $j$ (computed from the training range) and
$\sigma>0$ is an input-scale factor (default $\sigma=1$). For $d=1$ and $\sigma=1$ this is $e^{-ixX}$, which equals $e^{-ixZ}$ (with $S=\{\pm1\}$) up to a fixed basis change absorbed by the neighbouring $SU(2)$ blocks.
We computed the joint spectrum numerically (a sign vector $s\in\{\pm1\}^d$ belongs to $S$ iff the ordered product of eigenprojectors is non-zero):
$S=\{\pm1\}^d$ for $d=1,2,3,4,9$, hence $T=\{-2,0,2\}^d$, $\mathcal L=(2\mathbb Z)^d$, $m=\mathrm{rank}\,\mathcal L=d$ and
$|K_R|=(2R+1)^d$. So the encoding does not suffer from the rank deficiency of Example~\ref{ex:joint}, and the lattice is full-rank in every experiment.
\emph{Training.} Mean-squared-error loss on targets normalised to $[-1,1]$ (regression) or labels $\pm1$ (classification; prediction is $\mathrm{sign}\langle Z\rangle$);
full-batch Adam, learning rate decayed exponentially from $0.05$ to $0.003$; parameters initialised uniformly in $[-\pi,\pi]$.
\emph{Baseline.} A two-hidden-layer $\tanh$ MLP whose width is the largest with at most $P$ parameters, trained with the same loss, optimiser and epoch budget (learning rate $0.02\to0.001$ unless stated).
\emph{Spectrum check.} For $d=1$ and $R\in\{1,2,4,8\}$ we sampled random parameters, evaluated the simulated model on a grid and took an FFT:
exactly the $2R+1$ predicted frequencies $\{-2R,\dots,2R\}\cap2\mathbb Z$ were non-zero and none outside. For $d=2$ and $R=2,3$ all $(2R+1)^2$ frequencies of $K_R$ were
non-zero and none outside; for $R=1$ all non-zero frequencies lay in $K_1$ (8 of 9 were non-zero for the sampled parameters). This confirms Lemma~\ref{lem:fourier} and the
description of $K_R$ for the simulated architecture.

\subsection{QDRs as Regressors}

\subsubsection{Univariate Function Approximation}
Targets are normalised to $[-1,1]$ using their range on a dense grid of the domain (so a constant predictor has RMSE of order $0.5$);
$256$ training and $512$ test points are drawn uniformly from the domain, which is stated in Table~\ref{tab:uni}
(the domains are our choice). Figure~\ref{fig:univar} and Table~\ref{tab:uni} give the results for a 27-parameter ($R=8$) QDR trained for 400 epochs.

\begin{figure}[htbp]
  \centering
  \includegraphics[width=\linewidth]{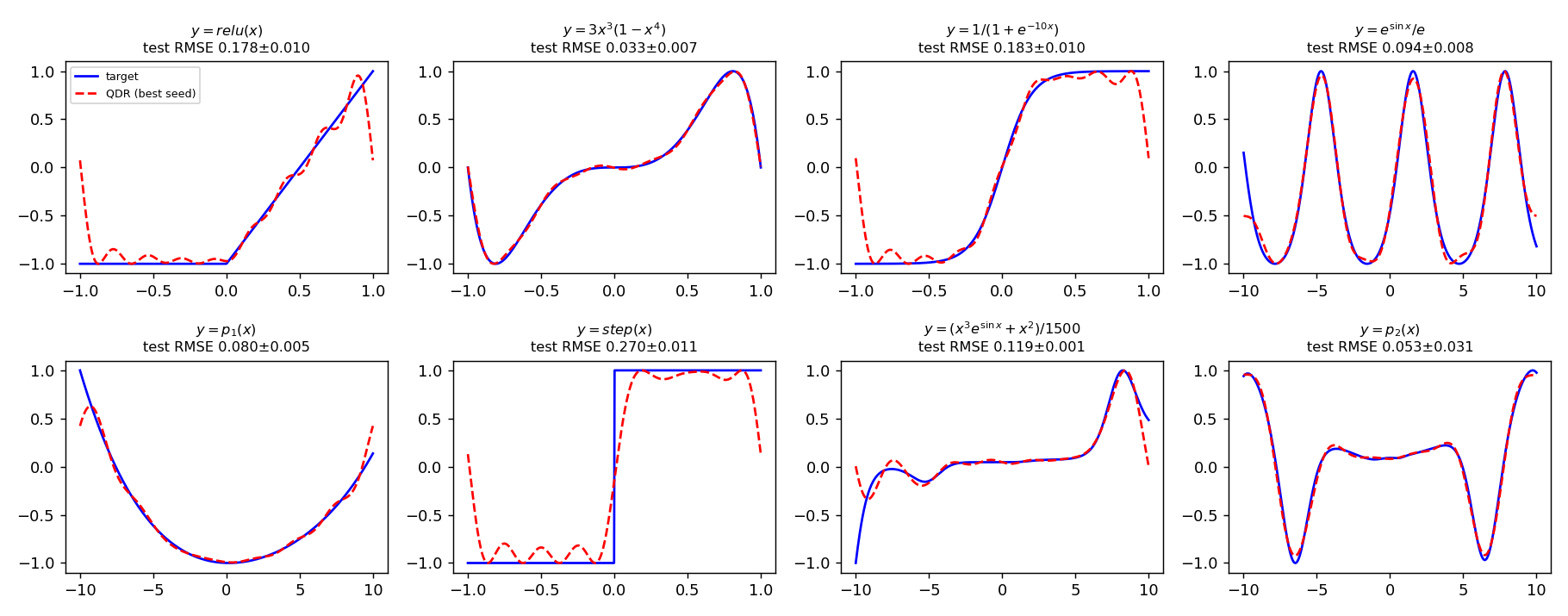}
  \caption{Univariate targets (blue) and the QDR fit (red dashed, best of 5 seeds) for a 27-parameter QDR with the default encoding ($\sigma=1$), 400 epochs. Targets are normalised to $[-1,1]$; the test RMSE shown is mean$\pm$std over seeds.}
  \label{fig:univar}
\end{figure}

\begin{table}[htbp]
\centering
\caption{Univariate regression, 27-parameter QDR ($R=8$), 400 epochs, 5 seeds, RMSE on targets normalised to $[-1,1]$. QDR with $\sigma=1$ (data domain $=$ one period of the model), QDR with $\sigma=0.5$ (half a period), and a 22-parameter MLP (the largest width $\le27$).}
\label{tab:uni}
\small
\begin{tabular}{@{}l c c c c c@{}}
\toprule
Target & Domain & QDR train ($\sigma{=}1$) & QDR test ($\sigma{=}1$) & QDR test ($\sigma{=}0.5$) & MLP test\\
\midrule
$\mathrm{relu}(x)$ & $[-1,1]$ & 0.195 & 0.178$\pm$0.010 & 0.030$\pm$0.006 & 0.050$\pm$0.004 \\
$3x^3(1-x^4)$ & $[-1,1]$ & 0.036 & 0.033$\pm$0.007 & 0.049$\pm$0.007 & 0.187$\pm$0.025 \\
$1/(1+e^{-10x})$ & $[-1,1]$ & 0.200 & 0.183$\pm$0.010 & 0.031$\pm$0.010 & 0.035$\pm$0.035 \\
$e^{\sin x}/e$ & $[-10,10]$ & 0.099 & 0.094$\pm$0.008 & 0.201$\pm$0.004 & 0.594$\pm$0.128 \\
$p_1(x)$ & $[-10,10]$ & 0.086 & 0.080$\pm$0.005 & 0.023$\pm$0.007 & 0.073$\pm$0.027 \\
$\mathrm{step}(x)$ & $[-1,1]$ & 0.290 & 0.270$\pm$0.011 & 0.248$\pm$0.010 & 0.198$\pm$0.044 \\
$(x^3e^{\sin x}+x^2)/1500$ & $[-10,10]$ & 0.126 & 0.119$\pm$0.001 & 0.081$\pm$0.014 & 0.151$\pm$0.016 \\
$p_2(x)$ & $[-10,10]$ & 0.049 & 0.053$\pm$0.031 & 0.102$\pm$0.016 & 0.372$\pm$0.112 \\
\bottomrule
\end{tabular}
\end{table}

With the default encoding the QDR has lower test RMSE than the MLP on four targets (the degree-7 polynomial, $e^{\sin x}/e$, $x^3e^{\sin x}+x^2$ and $p_2$) and clearly higher on three (ReLU, sigmoid, step);
it is comparable on $p_1$. The three targets on which the QDR is worse are exactly those whose values at the two ends of the domain differ and which are non-smooth or steep: with $\sigma=1$
the domain is one period of the model, so the model is forced to be periodic and the target's periodic extension has a jump, which a trigonometric polynomial resolves slowly. Halving the input scale
removes this wrap-around and reduces the test RMSE of ReLU from $0.178$ to $0.030$, of the sigmoid from $0.183$ to $0.031$ and of $p_1$ from $0.080$ to $0.023$; these are below the MLP's $0.050$, $0.035$ and $0.073$ respectively (for the sigmoid the difference to the MLP is within its seed variation).
On the targets that oscillate over several periods
($e^{\sin x}/e$ and $p_2$) the half-period encoding is worse ($0.201$ and $0.102$ versus $0.094$ and $0.053$), as expected since they need more of the available frequencies.
The step function remains poorly fitted by both models ($0.25$--$0.27$ for the QDR). We note that changing $\sigma$ alters both the periodic wrap-around and the effective frequency content, so this ablation is consistent with, but does not isolate, the wrap-around explanation.

\subsubsection{Effect of the Number of Parameters}
We study the target $y=\tanh(x^3)$ on $[-1.5,1.5]$ (normalised to $[-1,1]$) with a \emph{fixed} epoch budget for all parameter counts,
at 400 and at 1600 epochs, and compare with matched MLPs (Table~\ref{tab:sweep}, Figure~\ref{fig:tanhparams}).

\begin{table}[htbp]
\centering
\caption{Test RMSE for $\tanh(x^3)$ versus parameter count $P=3(R+1)$ with fixed epoch budgets, 5 seeds. The last two columns are the matched MLPs (their parameter counts are $13,22,33$ for $P\in[15,21]$, $[24,30]$, $\{36,45\}$).}
\label{tab:sweep}
\small
\begin{tabular}{@{}c c c c c c c@{}}
\toprule
$P$ & $R$ & QDR $\sigma{=}1$, 400 ep & QDR $\sigma{=}1$, 1600 ep & QDR $\sigma{=}0.5$, 400 ep & MLP 400 ep & MLP 1600 ep\\
\midrule
15 & 4 & 0.212$\pm$0.003 & 0.210$\pm$0.000 & 0.030$\pm$0.013 & 0.121 & 0.067 \\
18 & 5 & 0.205$\pm$0.004 & 0.200$\pm$0.004 & 0.020$\pm$0.008 & 0.121 & 0.067 \\
21 & 6 & 0.203$\pm$0.006 & 0.190$\pm$0.005 & 0.017$\pm$0.003 & 0.121 & 0.067 \\
24 & 7 & 0.190$\pm$0.009 & 0.178$\pm$0.002 & 0.017$\pm$0.003 & 0.068 & 0.039 \\
27 & 8 & 0.184$\pm$0.009 & 0.168$\pm$0.002 & 0.016$\pm$0.005 & 0.068 & 0.039 \\
30 & 9 & 0.173$\pm$0.009 & 0.157$\pm$0.002 & 0.018$\pm$0.004 & 0.068 & 0.039 \\
36 & 11 & 0.167$\pm$0.009 & 0.148$\pm$0.003 & 0.010$\pm$0.003 & 0.063 & 0.018 \\
45 & 14 & 0.153$\pm$0.006 & 0.132$\pm$0.005 & 0.010$\pm$0.002 & 0.063 & 0.018 \\
\bottomrule
\end{tabular}
\end{table}

\begin{figure}[htbp]
  \centering
  \includegraphics[width=\linewidth]{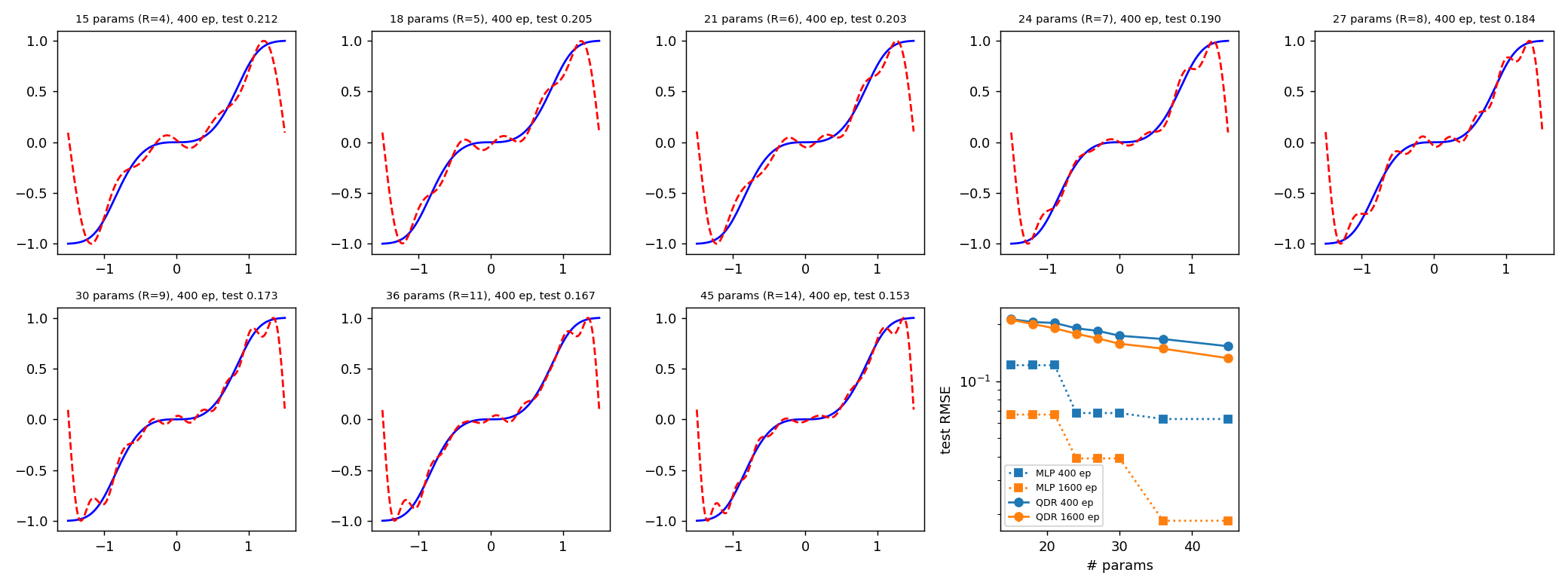}
  \caption{$\tanh(x^3)$ target (blue) and QDR fits (red dashed, best seed) for 15--45 parameters, all trained for 400 epochs with $\sigma=1$ (test RMSE in the panel titles; last two panels: test RMSE versus $P$ at 400 and 1600 epochs for the QDR and the matched MLP).}
  \label{fig:tanhparams}
\end{figure}

\begin{figure}[htbp]
  \centering
  \includegraphics[width=0.7\linewidth]{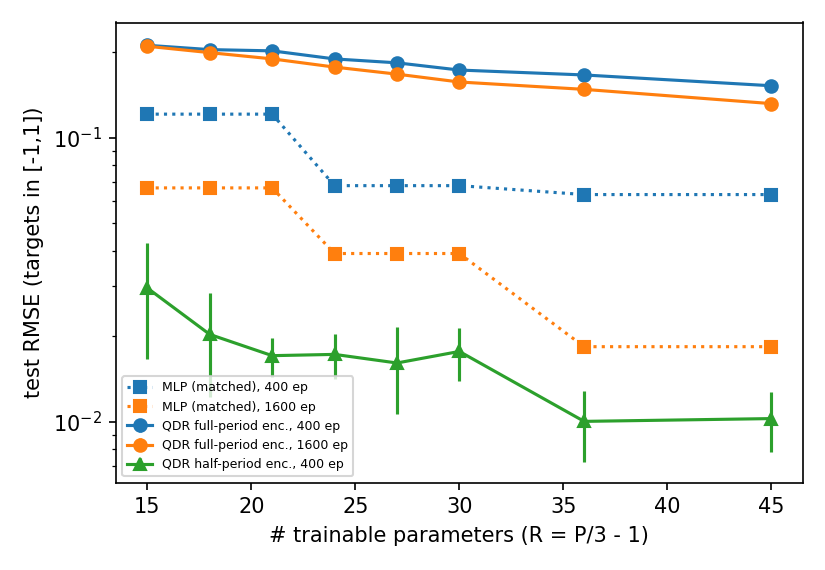}
  \caption{Test RMSE versus number of parameters for $\tanh(x^3)$ (mean$\pm$std over 5 seeds, log scale): QDR with $\sigma=1$ at 400 and 1600 epochs, QDR with $\sigma=0.5$, and matched MLPs.}
  \label{fig:tanhsummary}
\end{figure}

With the default encoding the error decreases steadily and slowly in $P$ (from $0.212$ at $15$ parameters to $0.153$ at $45$ at 400 epochs, and to $0.132$ at 1600 epochs); we do \emph{not} observe a plateau at about 27 parameters. The QDR with $\sigma=1$ is worse than the matched MLP at every $P$ (e.g.\ $0.153$ versus $0.063$ at $P=45$, 400 epochs).
The target is $\tanh(x^3)$ with end values $\mp1$, i.e.\ it has the same wrap-around problem as above, and with $\sigma=0.5$ the error is about an order of magnitude smaller ($0.030$ at $P=15$, $0.016$ at $P=27$, $\approx0.010$ at $P=36,45$),
decreasing up to about $P=36$ and then flat within seed noise at 400 epochs. This flattening is in the regime where Section~2.4 says it must come from coefficient constraints, optimisation or input scaling rather than from the spectrum; we did not
test which. In this single-qubit setting $|K_R|=2R+1\le P$, so the parameter count is not what limits the fit.

\subsubsection{Multivariate Function Approximation}
We implemented 27 Feynman equations from their closed forms. For each equation we use 1000 training and 1000 test points (non-finite samples removed) and normalise targets to $[-1,1]$ with the training range. For each equation the QDR's input scale $\sigma\in\{0.05,0.1,0.2,0.4\}$ and the MLP's learning rate $\in\{0.02,0.005\}$
were selected on a held-out validation set of 200 training points (2 seeds, 150 epochs); the models were then fitted on all 1000 training points for 300 epochs (3 seeds). The same $\sigma$ is used for $P=30$ and $P=45$.
``Mean-pred'' is the test RMSE of predicting the training mean. Each equation is evaluated at two parameter counts.

\begin{table*}[htbp]
\centering
\caption{Feynman equations: test RMSE (3 seeds), QDR versus parameter-matched MLP at $P=30$ ($R=9$) and $P=45$ ($R=14$); $d$ is the number of inputs and $\sigma$ the selected QDR input scale. MLP parameter counts: at most 30 / 45.}
\label{tab:DR_compact}
\scriptsize
\begin{tabular}{@{}l c c c c c c c@{}}
\toprule
Equation & $d$ & $\sigma$ & Mean-pred & QDR 30 & MLP 30 & QDR 45 & MLP 45\\
\midrule
I.6.2 & 2 & 0.2 & 0.397 & 0.060$\pm$0.018 & 0.068$\pm$0.017 & 0.037$\pm$0.006 & 0.034$\pm$0.005 \\
I.6.2b & 3 & 0.2 & 0.409 & 0.092$\pm$0.013 & 0.132$\pm$0.005 & 0.073$\pm$0.013 & 0.103$\pm$0.017 \\
I.9.18 & 9 & 0.05 & 0.280 & 0.259$\pm$0.002 & 0.059$\pm$0.008 & 0.258$\pm$0.001 & 0.059$\pm$0.008 \\
I.12.11 & 5 & 0.05 & 0.222 & 0.185$\pm$0.001 & 0.069$\pm$0.005 & 0.187$\pm$0.003 & 0.066$\pm$0.021 \\
I.13.12 & 5 & 0.05 & 0.174 & 0.148$\pm$0.001 & 0.057$\pm$0.024 & 0.147$\pm$0.003 & 0.055$\pm$0.025 \\
I.15.3x & 4 & 0.05 & 0.417 & 0.348$\pm$0.001 & 0.054$\pm$0.003 & 0.348$\pm$0.001 & 0.044$\pm$0.007 \\
I.16.6 & 3 & 0.1 & 0.519 & 0.068$\pm$0.007 & 0.069$\pm$0.003 & 0.055$\pm$0.006 & 0.069$\pm$0.016 \\
I.18.4 & 4 & 0.05 & 0.451 & 0.242$\pm$0.007 & 0.106$\pm$0.014 & 0.237$\pm$0.001 & 0.089$\pm$0.003 \\
I.26.2 & 2 & 0.2 & 0.363 & 0.046$\pm$0.006 & 0.060$\pm$0.004 & 0.038$\pm$0.007 & 0.050$\pm$0.015 \\
I.27.6 & 3 & 0.1 & 0.378 & 0.076$\pm$0.009 & 0.092$\pm$0.009 & 0.057$\pm$0.003 & 0.074$\pm$0.005 \\
I.29.16 & 4 & 0.1 & 0.413 & 0.293$\pm$0.002 & 0.292$\pm$0.083 & 0.290$\pm$0.001 & 0.169$\pm$0.091 \\
I.30.3 & 3 & 0.2 & 0.287 & 0.202$\pm$0.004 & 0.209$\pm$0.008 & 0.194$\pm$0.003 & 0.204$\pm$0.004 \\
I.30.5 & 3 & 0.1 & 0.254 & 0.073$\pm$0.007 & 0.074$\pm$0.007 & 0.075$\pm$0.003 & 0.065$\pm$0.005 \\
I.37.4 & 3 & 0.2 & 0.428 & 0.074$\pm$0.012 & 0.156$\pm$0.025 & 0.081$\pm$0.011 & 0.124$\pm$0.016 \\
I.40.1 & 6 & 0.05 & 0.368 & 0.247$\pm$0.001 & 0.118$\pm$0.012 & 0.249$\pm$0.002 & 0.107$\pm$0.005 \\
I.44.4 & 5 & 0.05 & 0.181 & 0.144$\pm$0.004 & 0.090$\pm$0.011 & 0.138$\pm$0.003 & 0.050$\pm$0.015 \\
I.50.26 & 4 & 0.05 & 0.303 & 0.252$\pm$0.001 & 0.241$\pm$0.009 & 0.250$\pm$0.000 & 0.237$\pm$0.008 \\
II.2.42 & 5 & 0.05 & 0.148 & 0.096$\pm$0.005 & 0.086$\pm$0.009 & 0.096$\pm$0.004 & 0.072$\pm$0.008 \\
II.6.15a & 5 & 0.1 & 0.182 & 0.145$\pm$0.005 & 0.076$\pm$0.011 & 0.144$\pm$0.003 & 0.079$\pm$0.013 \\
II.11.7 & 6 & 0.05 & 0.114 & 0.121$\pm$0.001 & 0.097$\pm$0.006 & 0.119$\pm$0.002 & 0.097$\pm$0.002 \\
II.11.27 & 4 & 0.05 & 0.295 & 0.121$\pm$0.002 & 0.084$\pm$0.012 & 0.121$\pm$0.004 & 0.076$\pm$0.003 \\
II.35.18 & 5 & 0.05 & 0.454 & 0.279$\pm$0.002 & 0.106$\pm$0.027 & 0.277$\pm$0.001 & 0.092$\pm$0.007 \\
II.36.38 & 8 & 0.05 & 0.185 & 0.178$\pm$0.002 & 0.088$\pm$0.009 & 0.178$\pm$0.002 & 0.085$\pm$0.004 \\
II.38.3 & 4 & 0.05 & 0.241 & 0.195$\pm$0.002 & 0.091$\pm$0.011 & 0.195$\pm$0.001 & 0.083$\pm$0.004 \\
III.9.52 & 6 & 0.1 & 0.215 & 0.193$\pm$0.006 & 0.161$\pm$0.008 & 0.193$\pm$0.002 & 0.156$\pm$0.013 \\
III.10.19 & 4 & 0.05 & 0.432 & 0.195$\pm$0.002 & 0.072$\pm$0.007 & 0.196$\pm$0.001 & 0.066$\pm$0.005 \\
III.17.37 & 3 & 0.2 & 0.317 & 0.149$\pm$0.068 & 0.122$\pm$0.018 & 0.062$\pm$0.006 & 0.065$\pm$0.014 \\
\bottomrule
\end{tabular}
\end{table*}

Summary (Table~\ref{tab:DR_compact}). The median test RMSE over the 27 equations is 0.149 (QDR) versus 0.090 (MLP) at $P=30$ and 0.147 versus 0.076 at $P=45$; the median mean-predictor RMSE is 0.303.
The QDR beats the matched MLP on 8 of 27 equations at $P=30$ and on 7 at $P=45$, and \emph{all} of these wins are on equations with $d\le3$ inputs (there are 9 such equations); on the 18 equations with $d\ge4$ it never wins at $P=30$.
The QDR beats the mean predictor on 26 of 27 equations at $P=30$ (the exception is II.11.7) but for 16 of the 18 equations with $d\ge4$ its RMSE is more than half of the mean predictor's, i.e.\ it explains only part of the variance.
Train and test errors are close (median train RMSE 0.157 against test 0.149 at $P=30$), so the poor fits are \emph{under-fitting}, not over-fitting. For 15 of 27 equations the validation search selected the smallest scale in the grid ($\sigma=0.05$),
so the grid was too narrow at that end, and for those the test error is practically identical at $P=30$ and $P=45$ (e.g.\ I.9.18: $0.259$ and $0.258$), suggesting the model is then confined to a very low-frequency regime and extra depth is not used.
 The total number of phases in $K_R$ grows as $(2R+1)^d$ while $P$ grows linearly, so for $d\ge4$ the parameter budget is far below $|K_R|$ (Remark in Section~2.3); optimisation of a single-qubit circuit with $R\,d$ encoding gates is another candidate. We did not test either.
In a preliminary run with $\sigma=1$ for all equations the QDR's test RMSE was above the mean-predictor level on most of the equations it covered, which shows how sensitive the QDR is to the encoding scale.

\subsection{QDRs as Classifiers}
Data (generated 400 training and 1000 test points): XOR (uniform on $[-1,1]^2$, label $=\mathrm{sign}(x_1x_2)$), two Gaussian blobs (means $\pm0.6\,(1,1)$, standard deviation $0.35$) and a two-arm spiral with noise $0.05$. Inputs are rescaled to $[-\pi/2,\pi/2]$ using the training range.
The ``3-layer, 9-parameter'' models are $R=2$; the spiral model has 15 trainable blocks ($R=14$, 45 parameters). Training is 500 epochs (XOR, blobs) or 800 (spiral), 5 seeds. We also vary the input scale $\sigma$ on the spiral.

\begin{table}[htbp]
\centering
\caption{Classification accuracy (mean$\pm$std over 5 seeds). The MLP has the largest width with at most $P$ parameters (reported in the column $P$ as the actual count).}
\label{tab:cls}
\small
\begin{tabular}{@{}l c c c@{}}
\toprule
Model & $P$ & Train acc. & Test acc.\\
\midrule
XOR, QDR & 9 & 0.90$\pm$0.04 & 0.84$\pm$0.04 \\
XOR, MLP & 7 & 0.68$\pm$0.01 & 0.66$\pm$0.02 \\
Gaussian blobs, QDR & 9 & 0.92$\pm$0.01 & 0.93$\pm$0.00 \\
Gaussian blobs, MLP & 7 & 1.00$\pm$0.01 & 0.99$\pm$0.01 \\
Spiral, QDR (scale 0.5) & 45 & 0.99$\pm$0.00 & 1.00$\pm$0.00 \\
Spiral, QDR (scale 1.0) & 45 & 0.99$\pm$0.01 & 0.97$\pm$0.00 \\
Spiral, QDR (scale 2.0) & 45 & 0.81$\pm$0.01 & 0.74$\pm$0.01 \\
Spiral, MLP & 37 & 0.84$\pm$0.13 & 0.85$\pm$0.12 \\
\bottomrule
\end{tabular}
\end{table}

\begin{figure}[htbp]
  \centering
  \includegraphics[width=\linewidth]{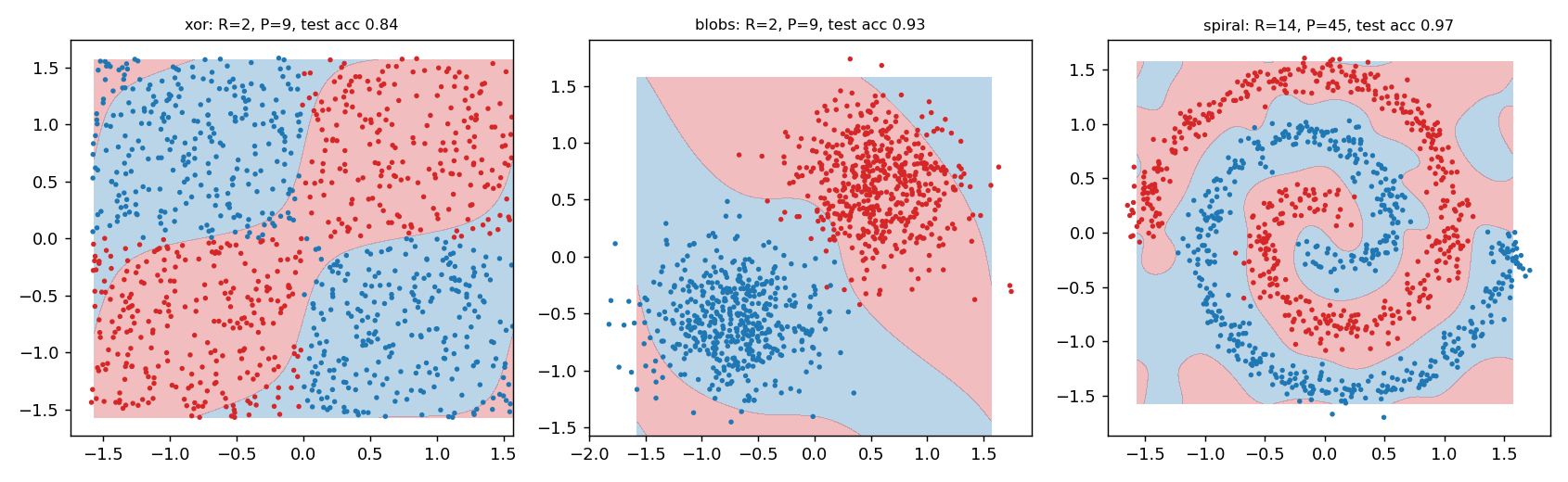}
  \caption{Decision regions of the best-of-5 QDR ($\sigma=1$) on the test points: XOR ($P=9$), Gaussian blobs ($P=9$) and spiral ($P=45$).}
  \label{fig:cls}
\end{figure}

On XOR the 9-parameter QDR reaches $0.84$ test accuracy, against $0.66$ for a 7-parameter MLP; on the Gaussian blobs it reaches $0.93$, below the MLP's $0.99$.
On the spiral, the 45-parameter QDR reaches $0.973$ test accuracy at $\sigma=1$ and $0.998$ at $\sigma=0.5$, while it drops to $0.738$ at $\sigma=2$; a 37-parameter MLP with the same optimiser reaches $0.847\pm0.123$ (high variance across seeds, with no tuning of its learning rate on this task).
The result thus depends strongly on the encoding scale. With larger $\sigma$ the data covers more periods of the model and the fit degrades; with $\sigma\le1$ the QDR fits the spiral well.
Because the MLP baseline was not tuned beyond its default learning rate on the classification tasks, we do not claim a superiority of either model on them.

\section{Conclusion}

We gave an explicit description of the Fourier spectrum of multivariate reuploading models:
$K_R$ is the ball of radius $R$ in the word metric of the difference set $T=S-S$ on the lattice
$\mathcal L$. The direction structure (the subspace $V$, of dimension $m=\mathrm{rank}\,\mathcal L$) is
fixed at $R=1$ and determines which functions are representable at any depth, whereas the spectrum itself grows without
bound, with maximal frequency linear and size $\Theta(R^m)$ in $R$. This gives a capacity bound of order
$\sqrt{|K_R|/n}$. Coefficient constraints and parameter counts, not the spectrum alone, determine which functions in
$K_R$ can actually be represented, and characterising them is open.

Our experiments confirm numerically the predicted spectra of the simulated architecture and show that, for single-qubit QDRs, (i) the input scale relative to the model's period has a larger effect on accuracy than depth in the ranges tested,
(ii) the QDR can beat a parameter-matched MLP on smooth or periodic univariate targets and on the spiral classification task, but is worse on targets with a periodic-extension jump unless the scale is reduced, and (iii) on multivariate regression it matches or beats the MLP only for $d\le3$
and is clearly worse for $d\ge4$, with an under-fitting rather than over-fitting failure. These are results for small single-qubit models with limited tuning (QDR scale grid edge-bound on 15 of 27 equations; MLP tuned only over two learning rates) and should not be read as general statements about QDRs.
Future work includes separating coefficient constraints from optimisation difficulty, wider scale and optimiser searches, multi-qubit encodings, and studying the reachable coefficient sets, possibly with residual connections or composing QDRs as units of larger variational models.

\bibliographystyle{splncs04}
\bibliography{main}

\end{document}